\documentclass[journal,letterpaper,twoside]{IEEEtran}

\usepackage[cmex10]{amsmath} 
\usepackage{amssymb}
\usepackage{amsthm}
\usepackage[colorlinks=true,linkcolor=black,citecolor=black,plainpages=false,pdfpagelabels]{hyperref}
\usepackage{cite}

\newtheorem{lem}{Lemma}
\newtheorem{thm}{Theorem}
\newtheorem{defin}{Definition}

\newcommand{\msf}{\mathsf}
\newcommand{\mrm}{\mathrm}

\newcommand{\mbb}{\mathbb}
\newcommand{\mcal}{\mathcal}

\newcommand{\bra}[1]{\langle #1|}
\newcommand{\ket}[1]{|#1 \rangle}

\newcommand{\ketbra}[1]{\ket{#1}\bra{#1}}

\newcommand{\ident}{\mathbb{I}}
\DeclareMathOperator{\tr}{Tr}

\newcommand{\mdag}{^{\dag}} 

\newcommand{\mmixed}{\Omega^{A\p}\mspace{-3mu}\otimes \rho^{E}}
\newcommand{\tra}{\tr_A}

\newcommand{\mpr}{{\mrm{Pr}}}
\newcommand{\mtr}[1]{\mathrm{Tr}\left(#1\right)}
\newcommand{\mtra}[1]{\mathrm{Tr}_{A}\left(#1\right)}
\newcommand{\p}{^{\prime}}

\newcommand{\mbI}{\mathbb{I}}
\newcommand{\mE}[1]{\mcal{E}(#1)}

\newcommand{\XV}{\frac{X^uZ^v}{\sqrt{d_A}}}
\newcommand{\XVd}{\frac{Z^vX^u}{\sqrt{d_A}}}
\newcommand{\ts}{\tilde{\rho}^{AE}}
\newcommand{\red}{ {\sigma^E}^{-\tfrac{1}{2}} }

\usepackage[utf8]{inputenc}
\begin{document}
\title{Quantum entropic security and approximate quantum encryption}
\author{Simon Pierre Desrosiers and Frédéric Dupuis
\thanks{F. Dupuis is with the Université de Montréal and McGill University. email: {\tt dupuisf@iro.umontreal.ca}}
\thanks{S. P. Desrosiers is with McGill University, email: {\tt simonpie@cs.mcgill.ca}}}

\maketitle

\IEEEpeerreviewmaketitle

\begin{abstract}
An encryption scheme is said to be entropically secure if an adversary whose min-entropy on the message is upper bounded cannot guess any function of the message. Similarly, an encryption scheme is entropically indistinguishable if the encrypted version of a message whose min-entropy is high enough is statistically indistinguishable from a fixed distribution. We present full generalizations of these two concepts to the encryption of quantum states in which the quantum conditional min-entropy, as introduced by Renner, is used to bound the adversary's prior information on the message. A proof of the equivalence between quantum entropic security and quantum entropic indistinguishability is presented. We also provide proofs of security for two different ciphers in this model and a proof for a lower bound on the key length required by any such cipher. These ciphers generalize existing schemes for approximate quantum encryption to the entropic security model.
\end{abstract}


\section{Introduction}

\PARstart{S}{emantic} security, whether it is computational, as introduced in \cite{GM84}, information theoretic in a classical setting, as introduced in \cite{RW2002} and \cite{DS2004}, or information theoretic in a limited quantum setting, as introduced in \cite{SPD2007}, contrasts the capabilities of two adversaries: one ($\msf{A}$) that has access to an encrypted version of the message, and another ($\msf{A}\p$) that does not.  Their abilities to predict a function on the initial message are compared.  Of course $\msf{A}\p$ seems to be at a tremendous disadvantage:  it has access to nothing but the prior distribution of the plain text, whereas $\msf{A}$ also has access to an encrypted version of the plain text and could potentially use imperfections in the encryption scheme to gain an advantage.  However, this can become a way to bound these imperfections: an encryption scheme is considered semantically secure if, for every adversary $\msf{A}$, there exists an $\msf{A}'$ that can predict every function on the plaintext almost as well as $\msf{A}$ without even having access to the encrypted message. This is a very strong security criterion, especially in the information theoretic setting.

Perhaps surprisingly, it is possible to construct semantically secure encryption schemes which, depending on their setting, make very few assumptions on $\msf{A}$ and yet do not require keys to be as long as the message. In the computational setting, Goldwasser and Micali \cite{GM84} had as a constraint that both $\msf{A}$ and $\msf{A}\p$ were probabilistic polynomial-time machines.  In their model, they could construct encryption schemes which, on all message distributions, would render $\msf{A}$ as useless as $\msf{A}\p$.  In the information theoretic setting, introduced by Russell and Wang \cite{RW2002}  and expanded upon by Dodis and Smith \cite{DS2004}, no computational limitation is imposed on $\msf{A}$ or $\msf{A}\p$.  In order to achieve significant key size reduction, a limit on the prior knowledge of $\msf{A}$ on the plain text space is assumed.  In fact, a lower bound on the min-entropy of the message space is assumed: the most probable message is not too probable. For this reason, this concept is called \emph{entropic security} in the context of information-theoretic security. In the quantum information theoretic setting, as introduced by Desrosiers \cite{SPD2007}, the exact same restriction on the min-entropy is imposed on $\msf{A}$, except that this time messages are quantum states which are further assumed to be unentangled with any quantum system that the adversary might possess. If these two restrictions are satisfied, one can construct encryption schemes for the quantum setting which have exactly the same key size as in the classical setting: for an $n$-qubit message which is assumed to have a min-entropy of at least $t$, then we need $n-t + \log(1/\varepsilon)$ bits of key to encrypt it securely (where $\varepsilon$ is a security parameter).

In this paper we remove one of those two restrictions.  Of course, the limit on the min-entropy of the adversary on the message space is hard to remove: it is the essence of entropic security.  However, it has to be modified in order to get robust definitions of security in the presence of entanglement between the sender and the adversary.  The notion of quantum conditional min-entropy as introduced by Renner in \cite{Renner2005} will be used to bound the prior ``knowledge'' of the adversary.  This new notion of min-entropy allows us to remove the no-entanglement restriction and replace it by something more general.  Indeed, if a state is not entangled, we have an implicit lower bound of zero on the conditional min-entropy, whereas in the general case, the conditional min-entropy of the adversary on an $n$-qubit system held by the sender ranges between $-n$ and $n$. It turns out that the key size remains the same in this model: for an $n$-qubit message about which the eavesdropper has a min-entropy of at least $t$, we still need a key of $n-t+\log(1/\varepsilon)$ bits. In the extreme case where we have no bound at all on the min-entropy, this reduces to $2n+\log(1/\varepsilon)$, which is in total agreement with the standard result of Ambainis, Mosca, Tapp and de Wolf \cite{amtw00}.  

Note that this generalizes the existing literature on approximate quantum encryption. In \cite{HLSW03}, Hayden, Leung, Shor and Winter considered the task of approximately encrypting quantum states assuming that the adversary is not entangled with the sender. They showed, using a randomized argument, that, while we need $2n$ bits of key to perfectly encrypt an $n$-qubit quantum message, there exists a scheme requiring $n+\log n + 2\log(1/\varepsilon) + O(1)$ bits of key.  Ambainis and Smith \cite{AS2004} then gave two explicit constructions of an approximate quantum encryption scheme under the same assumption requiring $n+2\log n + 2\log(1/\varepsilon)$ and $n+2\log(1/\varepsilon)$ bits of key respectively. Here we recover and generalize these results.

More recently, Fehr and Schaffner \cite{FS07} gave a classical encryption scheme which is entropically secure against an adversary that has access to quantum information about the classical message. Our work also generalizes this result: when our encryption schemes are applied to a classical message, the resulting ciphertext remains classical, and the proof of security still works against quantum adversaries.

We introduce our model and definitions in section \ref{modelanddefinitions} and show in section \ref{equivalence} that the two security definitions we give are equivalent. We also prove, in section \ref{encryptionschemes}, that two encryption schemes introduced by Ambainis and Smith \cite{AS2004} and by Dodis and Smith \cite{DS2004} (and generalized to the quantum world by Desrosiers \cite{SPD2007})  are still secure using this new definition and require the same amount of key as in the limited quantum model of \cite{SPD2007}.  Finally, in section \ref{minimum}, we generalize a proof of Dodis and Smith to show that an entropic scheme that can encrypt any $n$-qubit state having  a conditional min-entropy of at least $t$ requires at least $n-t-1$ bits of uniform key.

\section{Notation and preliminaries}
A quantum state $\rho$ is defined as a positive semidefinite operator of trace equal to 1 over some Hilbert space $\mcal{H}$ . By the spectral decomposition theorem, $\rho=\sum_{i}\gamma_{i}\ketbra{r_{i}}$, where the $\ket{r_{i}}$ form a basis for the space in which the quantum state lives and the $\gamma_{i}$ are non-negative real numbers that sum up to one.  This can be interpreted this way: if $\rho$ is measured in the basis $\{\ket{r_{i}}\}$, then it behaves as a source that will output with probability $\gamma_{i}$ the state $\ket{r_{i}}$.\footnote{For a thorough introduction to quantum information theory, see \cite{NC2000}}


The partial trace can be seen as a kind of inverse to the tensor product operation.  For any bipartite state $\rho^{AB}$, we have that $\rho^{B}=\mtra{\rho^{AB}}$; the normal interpretation for such an operator is that if a physical state $\rho^{AB}$ lives in the space $AB$  but one only has access to the system $B$ to measure the state, then the statistics obtained are in agreement with $\rho^{B}$.  The partial trace can be defined as:
\begin{equation}
\tra{(\rho^{AB})} \triangleq \sum_{i}(\bra{r_{i}}^{A}\otimes\mbI^{B})\rho^{AB}(\ket{r_{i}}^{A}\otimes\mbI^{B})
\end{equation}
where the vectors $\{\ket{r_{i}}\}$ form any orthonormal basis for the subspace $A$.  In fact, this is equivalent to doing a complete measurement of the $A$ subsystem followed by a loss of the result and of the $A$ subsystem; what is left in our hands is $\tra(\rho^{AB})$.

Throughout this paper, we will use superscripts for density matrices to indicate on which subsystems they are defined; for example, $\rho^{AB}$ is a density operator on the Hilbert space $\mathcal{H}_A \otimes \mathcal{H}_B$. By convention, when we omit certain subsystems from the superscript, we mean that we take the partial trace over the subsystems that are absent; i.e. $\rho^B = \tra \rho^{AB}$. We will refer to the dimension of the Hilbert space $\mathcal{H}_A$ by $d_A$.

We will use as our main distance measure the trace distance which is defined as 
\begin{equation} 
\left\| \rho - \sigma \right\|_{1} \triangleq \mtr{\left|\rho - \sigma \right|},
\end{equation}
where $\left| A \right|$ is defined as $\sqrt{A\mdag A}$, which is simply $\sum_{i}\left|\alpha_{i}\right|\ketbra{a_{i}}$ for a Hermitian operator $A=\sum_{i}\alpha_{i}\ketbra{a_{i}}$.  As \cite{helstrom1969} and chapter 9 in \cite{NC2000} tell us, for any two states $\rho$ and $\sigma$ there exists an optimal adversary which can distinguish between them with probability $\frac{1}{2} + \frac{1}{4} \|\rho-\sigma\|_{1}$; no adversary can do better.

Another useful distance measure is known as the \emph{fidelity}: given two density operators $\rho$ and $\sigma$, their fidelity $F(\rho, \sigma)$ is defined as $\left\| \sqrt{\rho} \sqrt{\sigma} \right\|_1$. If $\sigma$ is a pure state $\ketbra{\psi}$, this is equal to $\sqrt{\bra{\psi}\rho\ket{\psi}}$.

We will also frequently make use of operator inequalities: given two Hermitian operators $A$ and $B$, we will say that $A \geqslant B$ iff $A-B$ is positive semidefinite.

Also, we denote by $a\|b$ the concatenation of the bit strings $a$ and $b$. $X^a$, where $a=a_1 \cdots a_n$ is an $n$-bit string, means $X^a = X^{a_1} \otimes X^{a_2} \otimes \cdots \otimes X^{a_n}$.  We shall also write $\mcal{L(H)}$ for the space of linear operators on the Hilbert space $\mcal{H}$.  Finally, we denote by $a\odot b$ the inner product modulo $2$ of the strings $a$ and $b$: $\sum_{i}a_{i}b_{i}\mod{2}$.

\section{Model and definitions}\label{modelanddefinitions}

Entropic security as introduced by Russell and Wang \cite{RW2002} and generalized by Dodis and Smith \cite{DS2004} uses the definition of classical min-entropy to represent the adversary's knowledge on the sender's message space.  Let $M$ be a random variable over the message  space $\msf{M}$ and let $M$ take value $m$ with probability $p_{m}$.  Then the min-entropy of $M$, written $H_{\infty}(M)$ is defined to be $-\log{\max_{m}(p_{m})}$.

Desrosiers introduced in \cite{SPD2007} a quantum version of these security definitions for the case where the eavesdropper and the sender are neither entangled nor correlated. In this setting, a message $\sigma_{i}$ is chosen at random with probability $p_{i}$ in a valid \emph{interpretation} $\{(p_{i},\sigma_{i})\}$ of a state $\rho^{A}=\sum_{i}p_{i}\sigma_{i}$.  Here the adversary's a priori uncertainty is quantified by the quantum min-entropy, $H_{\infty}(\rho^{A})=-\log{\max_{j}\gamma_{j}}$  where $\sum \gamma_{j}\ketbra{j}$ is the spectral decomposition of $\rho^{A}$.  The joint system of the sender and the adversary was considered to contain no correlations: i.e. $\rho^{AE}=\sigma^{A}\otimes \tau^{E}$, where $E$ represents the eavesdropper's system.

In this paper, we shall show that we can fully generalize these security definitions to the quantum setting, where no assumption on the entanglement between the sender and the adversary is made.  The only restriction on the adversary will be quantified by the following definition introduced by Renner (see \cite{Renner2005}) in his proof that the BB84 scheme, the original quantum key distribution protocol, is secure in the most general setting.  We shall make no other assumption on the sender-eavesdropper system than the eavesdropper's conditional min-entropy.

\begin{defin}[Quantum conditional min-entropy]\label{ConMinEntropy}
For any quantum state $\rho^{AE}$ shared between the eavesdropper and the sender, we define the conditional min-entropy of $A$ given $E$ as
\begin{equation*}
H_{\infty}(A|E)_{\rho} = -\log \min_{\sigma^E} \min \left\{ \lambda : \lambda \ident^A \otimes \sigma^E \geqslant \rho^{AE} \right\}
\end{equation*}
where $\sigma^E$ ranges over all normalized density operators over $\mathcal{H}_E$.
\end{defin}

According to \cite{cond-min-ent}, we can express the quantum conditional min-entropy as
\begin{equation*}
2^{-H_{\infty}(A|E)_{\rho}} = d_A \max_{\mathcal{E}}\bra{\Phi} \mathcal{E}(\rho^{AE}) \ket{\Phi}
\end{equation*}
where the maximization is taken over all CPTP maps $\mathcal{E} : \mathcal{L}(\mathcal{H}_E) \rightarrow \mathcal{L}(\mathcal{H}_{A'})$ and where $\mathcal{H}_{A'} \cong \mathcal{H}_{A}$.

One can prove a few properties about conditional min-entropy which will be handy later on. First, this lemma:
\begin{lem}\label{separable}
	Let the joint state of the sender and the adversary be $\rho^{AE}=\rho^{A}\otimes\rho^{E}$, then $H_{\infty}(A|E)_{\rho}=H_{\infty}(A)_{\rho}$.
\end{lem}
\begin{proof}
	\begin{align*}
		2^{-H_{\infty}(A|E)_{\rho}} &= \min_{\sigma^E} \min \left\{ \lambda : \lambda \ident^A \otimes \sigma^E \geqslant \rho^A \otimes \rho^E \right\}\\
		&\geqslant \min \left\{ \lambda : \lambda \ident^{A} \geqslant \rho^A \right\}\\
		&= \min \left\{ \lambda : \lambda \ident^{A} \otimes \rho^E \geqslant \rho^A \otimes \rho^E\right\}\\
		&\geqslant \min_{\sigma^E} \min \left\{ \lambda : \lambda \ident^A \otimes \sigma^E \geqslant \rho^A \otimes \rho^E \right\}.
	\end{align*}
	Since the first and last lines are the same, the two inequalities are, in fact, equalities, and hence $2^{-H_{\infty}(A|E)_{\rho}} = \min\left\{ \lambda : \lambda \ident^A \geqslant \rho^A \right\} = 2^{-H_{\infty}(A)_{\rho}}$.
\end{proof}
We can conclude from this lemma that if the sender and the adversary are not correlated, then the earlier results of \cite{SPD2007} can be used.

Furthermore, König, Renner and Schaffner \cite{cond-min-ent} show that for a state of the form $\rho^{AE} = \sum_i p_i \ketbra{i}^A \otimes \rho_i^E$ (i.e. $A$ holds classical information and $E$ holds a quantum state containing partial information on $A$), the quantum conditional min-entropy $H_{\infty}(A|E)_{\rho}$ characterizes Eve's optimal probability of guessing $A$ by measuring $E$:
\begin{equation*}
	p_{\mathrm{guess}} = 2^{-H_{\infty}(A|E)_{\rho}}
\end{equation*}

Note also that if the $A$ and $E$ systems are in a maximally entangled state $\sum_{i=1}^{d}\frac{1}{\sqrt{d}}\ket{i}^{A}\ket{i}^{E}$, where $n=\log d$, then 
\begin{equation} H_{\infty}(A|E)_{\rho} = -n.
\end{equation}
Hence, the quantum conditional min-entropy ranges from $-n$ to $n$ for an $n$-qubits system and, as is the case with the von Neumann conditional entropy, negative values arise from purely quantum effects.


In our model, we will consider a protocol to be secure if the adversary is incapable of obtaining classical information about the message encoded in any basis. We will therefore model the adversary as a POVM on the encrypted message together with the adversary's side information. Since entropic security, even in the classical case (see \cite{DS2004}), does not have good composability properties (i.e. the security of the scheme does not necessarily imply that it can be securely embedded in a larger cryptographic protocol), we will not consider adversaries that keep quantum information without measuring it in the hopes of mounting a more effective attack later after having received more information. We are interested in the predictive capabilities of an adversary that was given $\mcal{E}(\sigma_{i})$ --- see below for the formal definition of a cipher $\mathcal{E}$ --- compared to those of an adversary that was not given such a state in predicting a function of $i$. Since our adversary is a POVM, we take its output to be a prediction of the function $f$.  We shall denote the random variable that is the output of $\msf{A}$ on any given state $\gamma$ by $\msf{A}(\gamma)$; that is, if $\{A_i\}_{i \in I}$ is the set of POVM elements associated with $\msf{A}$, then $\msf{A}(\gamma)$ is a random variable which takes the value $i$ with probability $\tr[A_i \gamma]$.

An encryption scheme $\mcal{E}$ is a set of superoperators $\{\mcal{E}_{k}\}$ indexed by a uniformly distributed key $k \in \{1,\ldots,K\}$ such that for each $k$ there exists an inverting operator $\mcal{D}_{k}$ such that for all $\rho^{AE}$, with probability one we have
\begin{equation}
(\mcal{D}_{k} \otimes \ident)( (\mcal{E}_{k} \otimes \ident)(\rho))=\rho.
\end{equation}
The view of the adversary is then $(\mcal{E} \otimes \ident)(\rho^{AE})\triangleq \frac{1}{K} \sum_{k=1}^K(\mcal{E}_{k} \otimes \ident)(\rho^{AE})$.  To simplify the notation, we will write $\mathcal{E}(\rho^{AE})$ instead of $(\mathcal{E} \otimes \ident)(\rho^{AE})$ from now on.  Note that in general, $\mcal{E}$ maps systems on space $AE$ to systems on space $A\p E$; the dimension of $A\p$ could be larger than the dimension of $A$.

Both \cite{DS2004} and \cite{SPD2007} presented security definitions equivalent in their respective models to the following two security definitions.

Note that throughout this paper, we shall be mostly concerned with encryption schemes where the message to be sent consists of $n$ qubits; therefore $n=\log d_A$ from now on.  

\begin{defin}[Entropic Security]\label{EntropicSecurity}
	An encryption system $\mathcal{E}$ is $(t,\varepsilon)$-entropically secure if for all states $\rho^{AE}$ such that $H_{\infty}(\rho^{AE}|\rho^E) \geqslant t$, all interpretations $\{ (p_i, \sigma_i^{AE})\}$, all adversaries $\msf{A}$ and all functions $f$, there exists an $\msf{A}'$ such that we have: \footnote{One can also get an equivalent definition by using functions on the states $\sigma_i^{AE}$ rather than on the indices $i$.}
\begin{equation}
	\left| \Pr[\msf{A}(\mathcal{E}(\sigma_i^{AE})) = f(i)] - \Pr[\msf{A}'(\sigma_i^E) = f(i)] \right| \leqslant \varepsilon.
	\label{eqn:entropic-security}
\end{equation}
\end{defin}
Note that everywhere, we take probabilities over all $i$ and all randomness used by the adversaries and the cipher.
\begin{defin}[Entropic Indistinguishability]\label{EntropicIndistinguishability}
	An encryption system $\mcal{E}$ is $(t,\varepsilon)$-indistinguishable if there exists a state $\Omega^{A'}$ such that  for all states $\rho^{AE}$ such that $H_{\infty}(A|E)_{\rho} \geqslant t$ we have that:
\begin{equation}\label{indinstinguishability}
\left\|\mcal{E}(\rho^{AE})- \Omega^{A'} \otimes \rho^E \right\|_{1}<\varepsilon.
\end{equation}
\end{defin}

\section{Equivalence between the two security definitions}\label{equivalence}
This section will show that an encryption scheme which is entropically secure is entropically indistinguishable, and vice-versa, up to small variations in the $t$ and $\varepsilon$ parameters. Before presenting these proofs, however, we will need an additional definition and a technical lemma. The following variation on entropic security will prove to be useful in the sequel:

\begin{defin}[Strong entropic security]\label{StrongEntropicSecurity}
	An encryption system $\mathcal{E}$ is strongly $(t,\varepsilon)$-entropically secure if for all states $\rho^{AE}$ such that $H_{\infty}(\rho^{AE}|\rho^E) \geqslant t$, all interpretations $\{ (p_i, \sigma_i^{AE})\}$, all adversaries $\msf{A}$, and all functions $f$, we have
\begin{equation}
	\left| \Pr[\msf{A}(\mathcal{E}(\sigma_i^{AE}))\mspace{-5mu} =\mspace{-5mu} f(i)]\mspace{-4mu} -\mspace{-4mu} \Pr[\msf{A}(\mathcal{E}(\rho^A) \otimes \sigma_i^E)\mspace{-5mu} =\mspace{-5mu} f(i)] \right|\mspace{-3mu} \leqslant\mspace{-3mu} \varepsilon.
	\label{eqn:strong-entropic-security}
\end{equation}
\end{defin}
Note that in this case both uses of $\mcal{E}$ are independent. Strong $(t,\varepsilon)$-entropic security clearly implies regular $(t,\varepsilon)$-entropic security, since $\msf{A}$ used on $\sigma_i^E$ and an encrypted message independent of $\sigma_i^E$ (which can be prepared by Eve in her lab) is a valid choice for $\msf{A}'$. 

The following lemma says that one does not need to consider all possible functions, but one can restrict the analysis to predicates:
\begin{lem}\label{allfunctions}
Let $\rho^{AE}$ be a state, $\{ (p_i,\sigma_i^{AE}) \}$ be an interpretation, $\mcal{E}$ be a cipher, $f$ be a function and $\msf{A}$ be an adversary such that 
\begin{equation*}
	\left| \Pr[\msf{A}(\mathcal{E}(\sigma_i^{AE}))\mspace{-5mu} =\mspace{-5mu} f(i)]\mspace{-4mu} -\mspace{-4mu} \Pr[\msf{A}(\mathcal{E}(\rho^A) \otimes \sigma_i^E)\mspace{-5mu} =\mspace{-5mu} f(i)] \right|\mspace{-3mu} > \mspace{-3mu} \varepsilon.
\end{equation*}
then there exist an adversary $B$ and a predicate $h$ such that
\begin{equation*}
	\left| \Pr[\msf{B}(\mathcal{E}(\sigma_i^{AE}))\mspace{-5mu} =\mspace{-5mu} h(i)]\mspace{-4mu} -\mspace{-4mu} \Pr[\msf{B}(\mathcal{E}(\rho^A) \otimes \sigma_i^E)\mspace{-5mu} =\mspace{-5mu} h(i)] \right|\mspace{-3mu} >\mspace{-3mu} \frac{\varepsilon}{2}.
\end{equation*}\end{lem}
\begin{proof}
Let our predicate be a Goldreich-Levin predicate \cite{GL1989}, that is $h_{r}(x)=r\odot f(x)$.  Let $p=\mpr[\msf{A}(\mcal{E}({\sigma}_{i}^{AE}))=f({i})]$ and $q=\mpr[\msf{A}(\mcal{E}(\rho^{A})\otimes{\sigma}_{i}^{E})=f({i})]$.  Then we know that $|p-q|> \varepsilon$.  Let us compute
 \begin{equation}\begin{split}\label{allfunction1}
E = \Big|\mbb{E}_{r}\big[\mrm{Pr}[r\odot \msf{A}(\mcal{E}({\sigma_{i}})) = &h_{r}({i})]\\
 -\mrm{Pr}[r\odot \msf{A}(\mE{\rho^{A}}&\otimes \sigma_{i}^{E})=h_{r}({i}) ]  \big]\Big|,
 \end{split}\end{equation}
where the expectation is taken over all $r$ of adequate size.
We need two observations.  First, when $\msf{A}$ predicts correctly, then $p=\mpr[r\odot \msf{A}(\mcal{E}({\sigma}_{i}^{AE}))=h_{r}({i})]$.  Second, when $\msf{A}$ does not predict correctly, the probability that $r\odot \msf{A}(\mcal{E}(\sigma_{i}))=h_{r}({i})$ is exactly one half.
Hence Equation \eqref{allfunction1} reduces to
\begin{equation}\begin{split}
E&=\left| 1\cdot p+\frac{1}{2}\cdot(1-p)-\left(1\cdot q+\frac{1}{2}\cdot (1-q)\right)\right|\\
&= \left| \frac{p-q}{2}\right| > \frac{\varepsilon}{2}.
\end{split}\end{equation}
Thus there exists at least one value $r$ such that the following is true:
 \begin{equation*}\begin{split}
\left| \Pr[r\odot\msf{A}(\mathcal{E}(\sigma_i^{AE}))\mspace{-5mu} \right.&=\mspace{-5mu} h_{r}(i)]\mspace{-4mu} \\
-&\left.\mspace{-4mu} \Pr[r\odot\msf{A}(\mathcal{E}(\rho^A) \otimes \sigma_i^E)\mspace{-5mu} =\mspace{-5mu} h_{r}(i)] \right|\mspace{-3mu} >\mspace{-3mu} \frac{\varepsilon}{2}.
\end{split}\end{equation*}
The lemma is proven if adversary $\msf{B}(\cdot)$ is defined, using this appropriate $r$, as $r\odot \msf{A}(\cdot)$.
\end{proof}

\begin{thm}\label{theorem1a}
	$(t-1,\varepsilon/2)$-entropic indistinguishability implies strong $(t,\varepsilon)$-entropic security for all functions.
\end{thm}
\begin{proof}
	We shall prove the contrapositive.  Suppose there exists an adversary $\msf{B}$, a state $\rho^{AE}$ such that $H_{\infty}(A|E)_{\rho} \geqslant t$, an interpretation $\left\{ (p_j, \sigma_j^{AE}) \right\}$ for $\rho^{AE}$ and a function $f$ such that
\begin{equation}
	\left| \Pr[\msf{B}(\mathcal{E}(\sigma_i^{AE}))\mspace{-5mu} =\mspace{-5mu} f(i)]\mspace{-4mu} -\mspace{-4mu} \Pr[\msf{B}(\mathcal{E}(\rho^A) \otimes \sigma_i^E)\mspace{-5mu} =\mspace{-5mu} f(i)] \right|\mspace{-4mu} >\mspace{-3mu} \varepsilon.
	\label{eqn:entropic-security-violation}
\end{equation}
Then we know from Lemma \ref{allfunctions} that there exists another adversary and a predicate $h$ such that strong $(t,\varepsilon/2)$-entropic security is violated. Let's call this adversary $\msf{A}$ and let us define the sets $E_0$ and $E_1$ as follows:
\begin{eqnarray}
	E_0 &=& \left\{ i | h(i) = 0 \right\}\\
	E_1 &=& \left\{ i | h(i) = 1 \right\}.
	\label{eqn:e0e1}
\end{eqnarray}
Define the following:
\begin{eqnarray*}
	r_0 &=& \sum_{i \in E_0} p_i,\\
	r_1 &=& \sum_{i \in E_1} p_i,\\
	\tau_0^{AE} &=& \frac{1}{r_0} \left( \sum_{i \in E_0} p_i \sigma_i^{AE} \right)\\
\tau_1^{AE} &=& \frac{1}{r_1} \left( \sum_{i \in E_1} p_i \sigma_i^{AE} \right).
\end{eqnarray*}
Note that $\rho^{AE}=r_{0}\tau_{0}^{AE}+r_{1}\tau_{1}^{AE}$.  Now, define the following states:
\begin{eqnarray}
	\tilde{\tau}_0^{AE} &=& r_0 \tau_0^{AE} + r_1 \rho^{A} \otimes \tau_1^E\\
	\tilde{\tau}_1^{AE} &=& r_1 \tau_1^{AE} + r_0 \rho^{A} \otimes \tau_0^E,
	\label{eqn:tautilde}
\end{eqnarray}
where, as usual, $\tau_{i}^{E}=\tra[\tau_{i}^{AE}]$.  We need the following lemma to finish the proof.
\begin{lem}
Assuming $H_{\infty}(A|E)_{\rho} \geqslant t$, we then have that both $H_{\infty}(A|E)_{\tilde{\tau}_0}$ and $H_{\infty}(A|E)_{\tilde{\tau}_1}$ are at least $t-1$.
\end{lem}
\begin{proof}
We have that
\begin{multline}\label{eqn:tau-min-entropy-bound}
	d_A\max_{\mathcal{E}} \bra{\Phi} \mathcal{E}(\tilde{\tau}_0^{AE}) \ket{\Phi}\\
\begin{split}
 &\leqslant r_0 d_A\max_{\mathcal{E}} \bra{\Phi} \mathcal{E}(\tau_0^{AE}) \ket{\Phi} + r_1 d_A\max_{\mathcal{E}} \bra{\Phi} \rho^A \otimes \mathcal{E}(\tau_1^E) \ket{\Phi}\\
 &\leqslant d_A\max_{\mathcal{E}} \bra{\Phi} \mathcal{E}(\rho^{AE}) \ket{\Phi} + d_A\max_{\mathcal{E}} \bra{\Phi} \rho^A \otimes \mathcal{E}(\rho^E) \ket{\Phi}\\
 &\leqslant 2^{-t} + d_A\max_{\mathcal{E}} \bra{\Phi} \rho^A \otimes \mathcal{E}(\rho^E) \ket{\Phi}
\end{split}
\end{multline}
We now bound the second term using the original definition of the conditional min-entropy:
\begin{multline}
\min_{\sigma^E} \min \left\{ \lambda : \lambda \ident^A \otimes \sigma^E \geqslant \rho^A \otimes \rho^E\right\}\\
\begin{split}
&\leqslant \min \left\{ \lambda : \lambda \ident^A \otimes \rho^E \geqslant \rho^A \otimes \rho^E \right\}\\
&= \min \left\{ \lambda : \lambda \ident^A \geqslant \rho^A \right\}\\
&\leqslant \min_{\sigma^E} \min \left\{ \lambda : \lambda \ident^A \otimes \sigma^E \geqslant \rho^{AE}\right\}\\
&\leqslant 2^{-t}
\end{split}
\end{multline}
Substituting this into the last line of (\ref{eqn:tau-min-entropy-bound}) yields $H_{\infty}(A|E)_{\tilde{\tau}_0} \geqslant t-1$. Of course, an identical calculation yields the same result for $\tilde{\tau}_1^{AE}$.
\end{proof}

To finish the proof of Theorem \ref{theorem1a} , we want to show that $\msf{A}$ can distinguish $\mathcal{E}(\tilde{\tau}_0^{AE})$ from $\mathcal{E}(\tilde{\tau}_1^{AE})$ with probability strictly better than $1/2 + \varepsilon/4$. Let's denote by $\eta$ the probability that $\msf{A}$ will correctly distinguish $\mathcal{E}(\tau_0^{AE})$ from $\mathcal{E}(\tau_1^{AE})$ in an $r_0, r_1$ mixture, and by $\alpha$ the probability that $\msf{A}$ will correctly distinguish $\mathcal{E}(\rho^A) \otimes \tau_0^E$ from $\mathcal{E}(\rho^A) \otimes \tau_1^E$ in an $r_0$, $r_1$ mixture. Also assume without loss of generality that $\eta > \alpha$ (otherwise consider an adversary identical to $\msf{A}$ but which returns the opposite answer). Now assume that we feed it $\mathcal{E}(\tilde{\tau}_0^{AE})$ with probability $1/2$ and $\mathcal{E}(\tilde{\tau}_1^{AE})$ with probability $1/2$. Observe that this is exactly as if we gave it an $r_0,r_1$ mixture of $\mathcal{E}(\tau_0^{AE})$ and $\mathcal{E}(\tau_1^{AE})$ with probability $1/2$ and an $r_0,r_1$ mixture of $\mathcal{E}(\rho^A) \otimes \tau_0^E$ and $\mathcal{E}(\rho^A) \otimes \tau_1^E$ with probability $1/2$. We then have that the probability of distinguishing $\mathcal{E}(\tilde{\tau}_0^{AE})$ from $\mathcal{E}(\tilde{\tau}_1^{AE})$ using $\msf{A}$ is
\begin{equation*}
	\frac{1}{2} \eta + \frac{1}{2}(1-\alpha) = \frac{1}{2} + \frac{1}{2}(\eta - \alpha)
\end{equation*}
since the correct answer is reversed for $\mathcal{E}(\rho^A) \otimes \tau_0^E$ and $\mathcal{E}(\rho^A) \otimes \tau_1^E$.

But by the assumption that $\msf{A}$ violates entropic security, we know that
\begin{equation*}\begin{split}
	\eta - \alpha &= \Pr[\msf{A}(\mathcal{E}(\tau_i^{AE}))\mspace{-2mu} =\mspace{-2mu} i] - \Pr \left[\msf{A}\left(\mathcal{E}(\rho^A) \otimes \tau_i^E\right)\mspace{-2mu} =\mspace{-2mu} i\right]\\ &> \varepsilon/2.
\end{split}\end{equation*}

Hence, the probability of distinguishing $\mathcal{E}(\tilde{\tau}_0^{AE})$ from $\mathcal{E}(\tilde{\tau}_1^{AE})$ is at least $1/2 + \varepsilon/4$, which implies that for all ${\Omega^{A}}\p$ we have:
\begin{align*}
	\varepsilon &< \left\| \mathcal{E}(\tilde{\tau}_0^{AE}) - \mathcal{E}(\tilde{\tau}_1^{AE}) \right\|_1 \\
	&= \left\| \left(\mathcal{E}(\tilde{\tau}_0^{AE})\mspace{-2mu} - \mspace{-3mu} \mmixed\right) \mspace{-3mu} - \mspace{-3mu} \left(\mathcal{E}(\tilde{\tau}_1^{AE}) \mspace{-2mu} - \mspace{-3mu} \mmixed \right) \right\|_1 \\
	&\leqslant \left\| \mathcal{E}(\tilde{\tau}_0^{AE}) - \mmixed \right\|_1 + \left\| \mathcal{E}(\tilde{\tau}_1^{AE}) - \mmixed \right\|_1 \\
\end{align*}
and therefore either $\left\| \mathcal{E}(\tilde{\tau}_0^{AE}) - \mmixed \right\|_1 > \varepsilon/2$ or $\left\| \mathcal{E}(\tilde{\tau}_1^{AE}) - \mmixed \right\|_1 > \varepsilon/2$, which is a violation of  $(t-1,\varepsilon/2)$-indistinguishability.

\end{proof}

\begin{thm}
	$(t,\varepsilon)$-entropic security implies $(t-1,6\varepsilon)$-indistinguishability as long as $t \leqslant n - 1$.
\end{thm}
\begin{proof}
	We will prove the contrapositive. Let $\mcal{E}(\mbI/d_{A})=\Omega^{A\p}$  and let $\rho^{AE}$ be a state such that $H_{\infty}(A|E)_{\rho} \geqslant t-1$ and $\left\| \mathcal{E}(\rho^{AE}) - \Omega^{A\p}\otimes \rho^{E} \right\|_1 > 6\varepsilon$.  Consider the following state
\begin{equation*}
	\widehat{\rho}^{AE} = \frac{1}{3} \rho^{AE} + \frac{2}{3} \frac{\mbI}{d_{A}}\otimes \rho^{E}.
\end{equation*}.

We show that $H_{\infty}(A|E)_{\widehat{\rho}} \geqslant t$:
\begin{multline*}
d_A\max_{\mathcal{E}} \bra{\Phi} \mathcal{E}(\widehat{\rho}^{AE}) \ket{\Phi}\\
\begin{split}
&\leqslant \frac{1}{3} d_A\max_{\mathcal{E}} \bra{\Phi} \mathcal{E}(\rho^{AE}) \ket{\Phi} + \frac{2}{3} d_A\max_{\mathcal{E}} \bra{\Phi} \left( \frac{\ident^A}{2^n} \otimes \mathcal{E}(\rho^E) \right) \ket{\Phi}\\
&\leqslant \frac{1}{3} 2^{-(t-1)} + \frac{2}{3} d_A\bra{\Phi} \left( \frac{\ident^{AE}}{2^n} \right) \ket{\Phi}\\
&\leqslant \frac{1}{3} 2^{-(t-1)} + \frac{2}{3} \cdot \frac{1}{2^n}\\
&= \frac{2}{3}\left(2^{-t} +\frac{1}{2^{n}}\right)\\
&\leqslant \frac{2}{3}\left(2^{-t} + \frac{2^{-t}}{2}\right)\\
&= 2^{-t}.
\end{split}
\end{multline*}

Since $\left\| \mathcal{E}(\rho^{AE}) - \mmixed \right\|_1 > 6\varepsilon$, we know that there exists an adversary that can distinguish $\mathcal{E}(\rho^{AE})$ from ${\Omega^{A}}\p \otimes \rho^E$ with probability at least $\frac{1}{2} + \frac{3}{2}\varepsilon$. Let's call this adversary $\msf{A}$, and let's assume that it gives the right answer with probability $\eta_1$ when it is given $\mathcal{E}(\rho^{AE})$ and with probability $\eta_2$ when it is given ${\Omega^{A}}\p \otimes \rho^E$. We then have $\frac{1}{2}(\eta_1 + \eta_2) > \frac{1}{2} + \frac{3}{2}\varepsilon$.

Now, consider the following interpretation of $\widehat{\rho}^{AE}$: 
\begin{equation}
	\widehat{\rho}^{AE} = \frac{1}{3} \sigma_1^{AE} + \frac{1}{3} \sigma_2^{AE}  + \frac{1}{3} \sigma_3^{AE}
	\label{eqn:interpretation}
\end{equation}
where $\sigma_1^{AE} = \rho^{AE}$ and $\sigma_2^{AE} = \sigma_3^{AE} = \frac{\ident}{d_A} \otimes \rho^E$. We shall show that $\msf{A}$ violates entropic security for $\widehat{\rho}^{AE}$, with this interpretation and the function $h(i) = i$.

First of all, it is clear that by having access only to Eve's system, no adversary can guess the value of $h$ with a probability greater than $1/3$. Let us now determine what $\msf{A}$ can do by having access to the encrypted version of $\widehat{\rho}^{AE}$. One possible strategy for $\msf{A}$ is to try to distinguish between $\mathcal{E}(\rho^{AE})$ and $\mmixed$ and return 1 when it gets $\mathcal{E}(\rho^{AE})$ and randomly return either 2 or 3 when it gets $\mmixed$. We then have:
\begin{eqnarray*}
	\Pr[\msf{A}(\mathcal{E}(\sigma_i^{AE})) = h(i)] &=& \frac{1}{3} \eta_1 + \frac{2}{3} \frac{\eta_2}{2}\\
	&=& \frac{1}{3}(\eta_1 + \eta_2)\\
	&>& \frac{1}{3}(1 + 3\varepsilon)\\
	&=& \frac{1}{3} + \varepsilon.
\end{eqnarray*}

Finally we get that for all adversaries $\msf{A}'$,
\begin{equation*}
\left| \Pr[\msf{A}(\mathcal{E}(\sigma_i^{AE}))\mspace{-5mu} =\mspace{-5mu} h(i)]\mspace{-4mu} -\mspace{-4mu} \underbrace{\Pr[\msf{A'}(\sigma_i^E)\mspace{-5mu} =\mspace{-5mu} h(i)]}_{=\frac{1}{3}} \right|\mspace{-3mu} >\mspace{-3mu} \varepsilon
\end{equation*}
a violation of entropic security.
\end{proof}

\section{Two encryption schemes}\label{encryptionschemes}
Before presenting the ciphers, we will give some definitions and technical lemmas which will be used in the presentation of both encryption schemes.

First, we define the following shortcut for any matrix $\sigma^{AE}$:
\begin{equation}
M_{uv}^{\sigma} := \tr_A\left[ \left(\XVd \otimes \ident\right) \sigma^{AE} \right].
\label{eqn:def-muv}
\end{equation}
We also define
\begin{equation}
\ts := \rho^{AE} - \frac{\ident}{d_A} \otimes \rho^E.
\label{def-ts}
\end{equation}
for any state $\rho^{AE}$, where $\sigma^E$ is a state such that $\rho^{AE} \leqslant 2^{-H_{\infty}(A|E)_{\rho}} \ident^A \otimes \sigma^E$.

\begin{lem} \label{lem:decomposition-sigma}
	For every density matrix $\sigma^{AE}$, we have that $\sigma^{AE} = \sum_{uv} \XV \otimes M_{uv}^{\sigma}$.
\end{lem}
\begin{proof}
	Let $\{E_j\}$ be an orthonormal basis for $\mathcal{L}(\mathcal{H}_E)$. Since Pauli matrices form an orthonormal basis for $\mathcal{L}(\mathcal{H}_A)$, we have
	\begin{align}
		\sigma^{AE} &= \sum_{uvj} \XV \otimes E_j \tr\left[ \left(\XVd \otimes E_j\mdag\right) \sigma^{AE} \right]\\
		&= \sum_{uvj} \XV \otimes E_j \tr\left[ E_j\mdag M_{uv}^{\sigma} \right]\\
		&= \sum_{uv} \XV \otimes \left\{ \sum_j E_j \tr\left[ E_j\mdag M_{uv}^{\sigma} \right] \right\}\\
		&= \sum_{uv} \XV \otimes M_{uv}^{\sigma}.
	\end{align}
\end{proof}

We will also make use of the following lemma (Lemma 5.1.3 in \cite{Renner2005}):
\begin{lem}\label{sauveteur}
Let S be a Hermitian operator and let $\sigma$ be any positive definite operator. Then
\begin{equation*}
\|S\|_{1}\leqslant \sqrt{\mtr{\sigma}\mtr{S\sigma^{-1/2}S\sigma^{-1/2}}}.
\end{equation*}
\end{lem}

\subsection{A scheme based on $\delta$-biased sets}

In \cite{AS2004}, Ambainis and Smith introduced an approximate quantum encryption scheme based on $\delta$-biased sets. Here, we shall show that if $H_{\infty}(A|E)_{\rho} \geqslant t$, then the Ambainis-Smith scheme is $(t,\varepsilon)$-secure using $n - t + 2\log n + 2\log(\frac{1}{\varepsilon})$ bits of key, where $n$ is the logarithm of $d_{A}$ as usual.

\begin{defin}[$\delta$-biased set]\label{def:DeltaBiased}
	A set $S \subseteq \{0,1\}^n$ is said to be $\delta$-biased if and only if for every $s' \in \{0,1\}^n, s' \neq 0^n$, we have that $\left|\frac{1}{|S|}\sum_{s \in S} (-1)^{s \odot s'}\right| \leqslant \delta$.
\end{defin}

There exist several efficient constructions of $\delta$-biased sets (\cite{NN93,ABNNR92,AGHP92}); following \cite{DS2004}, we will use the one from \cite{AGHP92}, which yields sets of size $n^2/\delta^2$ (note that Dickinson and Nayak \cite{nayak-dickinson} improve this to $\leqslant 16/\delta^2$).

The Ambainis-Smith scheme consists of applying an operator at random from the set
\begin{equation*}
	\left\{ X^a Z^b : a\| b \in S \mbox{ and }|a|=|b|=n\right\}	
\end{equation*}
where $S$ is a $\delta$-biased set containing strings of length $2n$. The shared private key is used to index one of the operators. In other words, the encryption operator is
\begin{equation*}
	\mathcal{E}(\rho^{AE}) = \frac{1}{|S|}\sum_{a\|b \in S} (X^a Z^b \otimes \ident) \rho^{AE} (Z^b X^a \otimes \ident)
\end{equation*}

We shall now prove that this scheme is secure in our framework. The following lemma contains most of the proof, and the main theorem follows:

\begin{lem}\label{as-final-lemma}
For any state $\rho^{AE}$ with $H_{\infty}(A|E)_{\rho} \geqslant t$, we have that
\begin{equation}
	\left\| \mathcal{E}(\rho^{AE}) - \frac{\ident}{d_A} \otimes \rho^E \right\|_1 \leqslant \delta \sqrt{d_A 2^{-t}}
\end{equation}
\end{lem}
\begin{proof}
Let $\sigma^E$ be a state such that $\rho^{AE} \leqslant 2^{-t} \ident^A \otimes \sigma^E$ and write
\begin{multline}
	\left\| \mathcal{E}(\rho^{AE}) - \frac{\ident}{d_A} \otimes \rho^E \right\|_1\\
\leqslant \sqrt{d_A \tr\left[ \mathcal{E}(\ts)(\ident\otimes \red) \mathcal{E}(\ts) (\ident \otimes \red )\right]} \label{eqn:principale}
\end{multline}
This is due to Lemma \ref{sauveteur}, with $\sigma = \ident \otimes \sigma^E$; without loss of generality, we can assume that $\rho^E$ has full rank by considering $\mathcal{H}_E$ to be the support of $\rho^E$. We continue by applying Lemma \ref{lem:decomposition-sigma} on $\ts$:
\begin{equation}
	\ts = \sum_{uv} \XV \otimes M_{uv}^{\tilde{\rho}}
\end{equation}
and therefore
\begin{align}
	\mathcal{E}(\ts) &= \sum_{uv} \mathcal{E}\left( \XV \right) \otimes M_{uv}^{\tilde{\rho}}\\
	&= \sum_{uv} \alpha_{uv} \XV \otimes M_{uv}^{\tilde{\rho}}
\end{align}
where $\alpha_{uv} = \frac{1}{|S|} \sum_{a\| b \in S} (-1)^{v\| u \odot a\| b}$, and since $\tr[\ts] = 0$, we can neglect the term $v\| u = 0^{2n}$, and hence $|\alpha_{uv}| \leqslant \delta$.

We now compute the trace in (\ref{eqn:principale}) as follows:
\begin{multline}\label{eqn:monstre}
	\tr\left[ \mathcal{E}(\ts)(\ident\otimes \red) \mathcal{E}(\ts) (\ident \otimes \red )\right]\\
\begin{split}
	&= \tr\left[ \left( \sum_{uv} \alpha_{uv} \XV \otimes M_{uv}^{\tilde{\rho}} \right)\right.\\
	& \hspace{1cm}\left.	\left( \sum_{uv} \alpha_{uv} \XV \otimes \red M_{uv}^{\tilde{\rho}} \red \right) \right]\\
	&\stackrel{(a)}{=} \tr\left[ \left( \sum_{uv} \alpha_{uv} \XV \otimes M_{uv}^{\tilde{\rho}} \right)\right.\\
	& \hspace{1cm}\left. \left( \sum_{uv} \alpha_{uv} \XVd \otimes \red {M_{uv}^{\tilde{\rho}}}\mdag \red \right) \right]\\
	&\stackrel{(b)}{=} \tr\left[ \sum_{uv} \alpha_{uv}^2 \frac{\ident^A}{d_A} \otimes M_{uv}^{\tilde{\rho}} \red {M_{uv}^{\tilde{\rho}}}\mdag \red \right]\\
	&\stackrel{(c)}{\leqslant} \delta^2 \tr\left[ \sum_{uv} \frac{\ident^A}{d_A} \otimes M_{uv}^{\tilde{\rho}} \red {M_{uv}^{\tilde{\rho}}}\mdag \red \right]\\
	&\stackrel{(d)}{\leqslant} \delta^2 \tr\left[ \sum_{uv} \frac{\ident^A}{d_A} \otimes M_{uv}^{\rho} \red {M_{uv}^{\rho}}\mdag \red \right]\\
	&= \delta^2 \tr\left[ \left( \sum_{uv} \XV \otimes M_{uv}^{\rho} \right)\right.\\
	& \hspace{1cm} \left. \left( \sum_{uv} \XVd \otimes \red {M_{uv}^{\rho}}\mdag \red \right) \right]\\
	&\stackrel{(e)}{=} \delta^2 \tr\left[ \rho^{AE} \left( (\ident \otimes \red) \rho^{AE} (\ident \otimes \red) \right)\mdag \right]\\
	&= \delta^2 \tr\left[ \rho^{AE} (\ident \otimes \red) \rho^{AE} (\ident \otimes \red) \right]\\
	&\stackrel{(f)}{\leqslant} \delta^2 \tr\left[ \rho^{AE} 2^{-t} \ident^{AE} \right] = \delta^2 2^{-t}
\end{split}
\end{multline}
where
\begin{itemize}
	\item $(a)$ comes from the fact that $(\ident \otimes \red) \mathcal{E}(\ts) (\ident \otimes \red)$ is Hermitian, hence taking its adjoint leaves it unchanged;
	\item $(b)$ is true because terms in which the $u,v$ pairs are not the same in both sums disappear when we take the trace;
	\item $(c)$ because $\alpha^2_{uv} \geqslant \delta^2$ and every term in the sum has a nonnegative trace since $\tr[M_{uv}^{\tilde{\rho}} \red {M_{uv}^{\tilde{\rho}}}\mdag \red] = \tr[(\sigma^{-\tfrac{1}{4}} M_{uv}^{\tilde{\rho}} \sigma^{-\tfrac{1}{4}})(\sigma^{-\tfrac{1}{4}} M_{uv}^{\tilde{\rho}} \sigma^{-\tfrac{1}{4}})\mdag]$.
	\item $(d)$ is justified below;
	\item $(e)$ is due to Lemma \ref{lem:decomposition-sigma}; and
	\item $(f)$ comes from the fact that $\rho^{AE} \leqslant 2^{-t} \ident \otimes \sigma^E \Rightarrow (\ident \otimes \red) \rho^{AE} (\ident \otimes \red) \leqslant 2^{-t} \ident^{AE}$.
\end{itemize}

To justify $(d)$, we first observe that $M^{\rho}_{uv} = M^{\tilde{\rho}}_{uv} + M^{\tau}_{uv}$, where $\tau = \frac{\ident}{d_A} \otimes \rho^E$.  Hence,
\begin{multline}
	\tr\left[ \sum_{uv} \frac{\ident^A}{d_A} \otimes M_{uv}^{\rho} \red {M_{uv}^{\rho}}\mdag \red \right]\\
	\begin{split}
		&= \tr\left[ \sum_{uv} \frac{\ident^A}{d_A} \otimes M_{uv}^{\tilde{\rho}} \red {M_{uv}^{\tilde{\rho}}}\mdag \red \right]\\	
		&+ \tr\left[ \sum_{uv} \frac{\ident^A}{d_A} \otimes M_{uv}^{\tilde{\rho}} \red {M_{uv}^{\tau}}\mdag \red \right]\\	
		&+ \tr\left[ \sum_{uv} \frac{\ident^A}{d_A} \otimes M_{uv}^{\tau} \red {M_{uv}^{\tilde{\rho}}}\mdag \red \right]\\	
	&+ \tr\left[ \sum_{uv} \frac{\ident^A}{d_A} \otimes M_{uv}^{\tau} \red {M_{uv}^{\tau}}\mdag \red \right].
	\end{split}
\end{multline}

Step $(d)$ then follows when we combine this with the observation that $M^{\tau}_{00} = \rho^E$, $M^{\tau}_{uv} = 0$ if $uv \neq 00$, and $M_{00}^{\tilde{\rho}} = 0$: the first sum is what we want to bound; the two sums in the middle evaluate to the zero matrix; and in the last sum, only the 00 term remains, which clearly has a positive trace.

Substituting the end result of (\ref{eqn:monstre}) in (\ref{eqn:principale}), we obtain:
\begin{equation}
	\left\| \mathcal{E}(\rho^{AE}) - \frac{\ident}{d_A} \otimes \rho^E \right\|_1 \leqslant \delta \sqrt{d_A 2^{-t}}.
\end{equation}
\end{proof}

The main theorem now easily follows:

\begin{thm}
	If $H_{\infty}(A|E)_{\rho} \geqslant t$, then the Ambainis-Smith scheme is $(t,\varepsilon)$-secure using $n - t + 2\log n + 2\log(\frac{1}{\varepsilon}) + 2$ bits of key, where $n=\log{d_{A}}$.
\end{thm}
\begin{proof}
	If we choose $\delta = \varepsilon/2^{(n - t)/2}$ and construct $S$ using the method of \cite{AGHP92} such that $|S| = (2n)^2/\delta^2$, by Lemma \ref{as-final-lemma} we obtain $\left\| \mathcal{E}(\rho^{AE}) - \frac{\ident}{d_{A}} \otimes \rho^E \right\|_1 \leqslant \varepsilon$ using $n - t + 2\log n + 2 \log(\frac{1}{\varepsilon}) + 2$ bits of key.
\end{proof}

\subsection{A scheme based on XOR-universal functions}
Our second scheme based on XOR-universal functions can be considered as a quantum version of the scheme given in \cite{DS2004}.  This scheme can also be viewed as a generalization of the second scheme of \cite{AS2004}.

\begin{defin}\label{def69}
Let  $\msf{H}_{n}=\{h_{i}\}_{i\in I}$ be a finite family of functions from $n$-bit strings to $n$-bit strings. We say the family $\msf{H}_{n}$ is strongly-XOR-universal if for all $n$-bit strings $a$, $x$, and $y$ such that $x \neq y$ we have
$$\mpr_{i}[h_{i}(x)\oplus h_{i}(y)=a] = \frac{1}{2^n}. $$
\end{defin}
where $i$ is distributed uniformly over $I$. The family proposed in \cite{DS2004} naturally possesses this property if one allows $i$ to be zero.

We define our second cipher as follows. Let $\msf{H}_{2n}$ be a strongly-XOR-universal family of functions. The encryption operator for the key $k$ is defined as
\begin{equation}
	\mcal{E}_{k}(\rho) = \frac{1}{|I|}\sum_{i \in I} \ketbra{i}^{A'} \otimes (X^aZ^b\otimes \mbI^{E})\rho^{AE} (Z^bX^a\otimes \mbI^{E})
\end{equation}
where $a\|b=h_{i}(k)$, $|a|=|b|=n$, $h_i \in \msf{H}_{2n}$ and $k$ is the secret key selected uniformly at random from a set $K \subseteq \{0,1\}^{2n}$. The overall cipher can be described by the superoperator
\begin{equation}
	\mathcal{E}(\rho) = \frac{1}{|I||K|} \sum_{i \in I, k \in K} \ketbra{i}^{A'} \otimes (X^a Z^b \otimes \ident^E) \rho^{AE} (Z^b X^a \otimes \ident^E).
\end{equation}
The structure of $K$ is irrelevant; only its cardinality matters for the security of the scheme. Not that this scheme is not length preserving since the ancillary system $A\p$ is part of the ciphertext. We now prove that this scheme is secure with the following theorem:

\begin{thm}\label{secondscheme}
	$\mathcal{E}$ is $(t,\varepsilon)$-indistinguishable if $\log |K|\geqslant n - t + 2\log(1/\varepsilon) $.  
\end{thm}

\begin{proof}
	To show that the cipher is $(t,\varepsilon)$-indistinguishable, we must show that for all states $\rho^{AE}$ such that $H_{\infty}(A|E)_{\rho} \geqslant t$,
	\begin{equation}
		\left\| \mathcal{E}(\rho^{AE}) - \frac{\ident^{AA'}}{|I|d_A} \otimes \rho^E \right\|_1 \leqslant \varepsilon.
	\end{equation}
	As in the proof of our other scheme, we use Lemma \ref{sauveteur} with $\sigma = \ident^{AA'} \otimes \rho^E$ to bound this:
\begin{multline}
	\left\| \mathcal{E}(\rho^{AE}) - \frac{\ident^{AA'}}{|I|d_A} \otimes \rho^E \right\|_1\\
\leqslant \sqrt{|I|d_A \tr\left[ \mathcal{E}(\ts)(\ident\otimes \red) \mathcal{E}(\ts) (\ident \otimes \red )\right]}. \label{eqn:principale2}
\end{multline}

To compute the trace in the above expression, we first express $\mathcal{E}(\tilde{\rho}^{AE})$ using Lemma \ref{lem:decomposition-sigma}:
\begin{align}
	\mathcal{E}(\tilde{\rho}^{AE}) &= \sum_{uv} \mathcal{E}\left( \frac{X^uZ^v}{\sqrt{d_A}} \right) \otimes M_{uv}^{\tilde{\rho}}\\
	&= \sum_{uvki} \alpha_{uvki} \ketbra{i} \otimes \frac{X^uZ^v}{\sqrt{d_A}} \otimes M_{uv}^{\tilde{\rho}}
\end{align}
where $\alpha_{uvki} = \frac{1}{|I||K|} (-1)^{v\|u \odot a\|b}$ where $a\|b = h_i(k)$.

We are now ready to evaluate the trace in (\ref{eqn:principale2}):
\begin{multline}\label{eqn:monstre2}
	\tr\left[ \mathcal{E}(\ts)(\ident\otimes \red) \mathcal{E}(\ts) (\ident \otimes \red )\right]\\
\begin{split}
	&= \tr\left[ \left( \sum_{uvki} \alpha_{uvki} \ketbra{i} \otimes \XV \otimes M_{uv}^{\tilde{\rho}} \right)\right.\\
	& \hspace{1cm}\left.	\left( \sum_{uvki} \alpha_{uvki} \ketbra{i} \otimes \XV \otimes \red M_{uv}^{\tilde{\rho}} \red \right) \right]\\
	&\stackrel{(a)}{=} \tr\left[ \left( \sum_{uvki} \alpha_{uvki} \ketbra{i} \otimes \XV \otimes M_{uv}^{\tilde{\rho}} \right)\right.\\
	& \hspace{1cm}\left. \left( \sum_{uvki} \alpha_{uvki} \ketbra{i} \otimes \XVd \otimes \red {M_{uv}^{\tilde{\rho}}}\mdag \red \right) \right]\\
&\stackrel{(b)}{=} \tr\left[ \sum_{uvkk'i} \alpha_{uvki} \alpha_{uvk'i} \frac{\ident^A}{d_A} \otimes M_{uv}^{\tilde{\rho}} \red {M_{uv}^{\tilde{\rho}}}\mdag \red \right]\\
	&\stackrel{(c)}{=} \frac{1}{|I||K|}\tr\left[ \sum_{uv} \frac{\ident^A}{d_A} \otimes M_{uv}^{\tilde{\rho}} \red {M_{uv}^{\tilde{\rho}}}\mdag \red \right]\\
	&\stackrel{(d)}{\leqslant} \frac{2^{-t}}{|I||K|}
\end{split}
\end{multline}
where
\begin{itemize}
	\item $(a)$ comes from the fact that $(\ident \otimes \red) \mathcal{E}(\ts) (\ident \otimes \red)$ is Hermitian, hence taking its adjoint leaves it unchanged;
	\item $(b)$ is true because terms in which the $u,v,i$ triples are not the same in both sums disappear when we take the trace. Taking the partial trace on the subsystem containing $\ketbra{i}$ then yields this.
	\item $(c)$ is justified below
	\item $(d)$ follows exactly the same argument as in equation block (\ref{eqn:monstre}) from line $(d)$ onwards.
\end{itemize}

We now justify step $(c)$. We first consider the terms of the sum in which $k \neq k'$. In the following, let $a\|b = h_i(k)$ and $c\|d = h_i(k')$. If $k \neq k'$, we have
\begin{multline}
	\sum_{i \in I} \alpha_{uvki} \alpha_{uvk'i}\\
	\begin{split}
		&= \sum_{i \in I}\frac{1}{|I|^2|K|^2} (-1)^{v\|u \odot a\|b} (-1)^{v\|u \odot c\|d}\\
		&= \sum_{i \in I}\frac{1}{|I|^2|K|^2} (-1)^{(v\|u \odot a\|b) \oplus (v\|u \odot c\|d)}\\
	&= \sum_{i \in I}\frac{1}{|I|^2|K|^2} (-1)^{v\|u \odot (a\|b \oplus c\|d)}.
	\end{split}
\end{multline}
However, by Definition \ref{def69}, $a\|b \oplus c\|d$ is uniformly distributed over all $2n$-bit strings when $i$ is chosen uniformly at random. This sum is therefore equal to zero whenever $u\|v \neq 0^{2n}$, and to $\frac{1}{|I||K|^2}$ when $u\|v = 0^{2n}$. However, we observe that $M_{00}^{\tilde{\rho}} = 0$, and hence those terms also disappear from the sum inside the trace.

To take care of the case where $k = k'$, it can easily be shown that $\alpha_{uvki}^2 = \frac{1}{|I|^2|K|^2}$. Summing over all $i$ and $k$, step $(c)$ follows.

Now, by hypothesis, we have $\log|K| \geqslant n - t + 2\log(1/\varepsilon)$, which can be transformed into $- \log|K| - t \leqslant \log(\varepsilon^2) - \log d_A$. Exponentiating both sides yields $\frac{2^{-t}}{|K|} \leqslant \frac{\varepsilon^2}{d_A}$. Combining this bound with (\ref{eqn:monstre2}) and substituting in (\ref{eqn:principale2}) concludes the proof.

\end{proof}

\section{Minimum requirement for the key length}\label{minimum}

We can generalize the proof for the lower bound on the key length found in \cite{DS2004} to the quantum world and the conditional min-entropy definition.

\begin{thm}
Any quantum encryption scheme which is $(t,\varepsilon)$-indistinguishable for inputs of $n$ qubits requires a key of length at least $n-t-1$ as long as $\varepsilon \leqslant 1/2$.
\end{thm}
\begin{proof}
	We prove this by constructing a state with conditional min-entropy $t$ which provably requires at least $n - t -1$ bits of key to be securely encrypted. Consider the state $\rho^{A  \hat{A} E} = \ketbra{\Phi^+}^{A  E} \otimes \frac{\ident^{\hat{A}}}{d_{\hat{A}}}$ where $\ket{\Phi^{+}}^{AE} = \sum_{i=1}^{d_A} \ket{i}^{A}\ket{i}^{E}$ is a maximally entangled state; Alice wants to send both $A $ and $\hat{A}$ to Bob securely. Furthermore, let $d_{A } = d_E = 2^{(n - t)/2}$ and $d_{\hat{A}} = 2^{(n + t)/2}$, hence $d_{A\hat{A}}=2^{n}$. It is easy to compute the conditional min-entropy of this state: 
\begin{align*}
H_{\infty}(A  \hat{A}|E)_{\rho} &= H_{\infty}(A |E)_{\ketbra{\Phi^+}} + H_{\infty}(\hat{A})_{\ident^{\hat{A}}/d_{\hat{A}}}\\
      &= -(n-t)/2 + (n+t)/2\\
      &= t.
\end{align*}
Now, it is clear that this state requires at least as much key to encrypt as $\ketbra{\Phi^+}^{A  E}$ alone, since one could securely encrypt $\ketbra{\Phi^+}^{A  E}$ using a protocol to encrypt $\rho^{A  \hat{A} E}$ by adding $(n+t)/2$ random qubits to the input state. However, as the following theorem proves, $\ketbra{\Phi^+}^{A  E}$ requires at least $(n - t) - 1$ bits of key to encrypt.
\end{proof}

\begin{thm}\label{thm:main-lower-bound}
Let $\mcal{E}^{\tilde{A} \rightarrow A}$ be a cipher such that for all states  $\rho^{\tilde{A}E}$, there exists some state $\Omega^{{A}}$ such that 
\begin{equation}
\left\|(\mcal{E}^{\tilde{A}}\otimes \mbI^{E})(\rho^{\tilde{A}E})-\Omega^{{A}}\otimes \rho^{E} \right\|_{1}\leqslant\varepsilon,
\end{equation}
then $\mcal{E}$ requires at least $2\log(d_{\tilde{A}})-1$ bits of key, or $2n-1$ bits of key for an $n$-qubit system, whenever $\varepsilon \leqslant 1/2$.
\end{thm}
Before proving this, we first need a technical lemma which says that by conditioning on a classical system, we cannot reduce the min-entropy by more than the dimension of the system:

\begin{lem}\label{lem:lower-bound-min-entropy-omega}
	Given a state $\omega^{AEK} = \sum_k p_k \omega^{AE}_k \otimes \omega^K_k$, we have that $H_{\infty}(E|AK)_{\omega} \geqslant H_{\infty}(E|A)_{\omega} - \log|K|$.
\end{lem}
\begin{proof}
	\begin{align*}
		2^{-H_{\infty}(E|AK)} &= \min_{\sigma^{AE}} \min \left\{ \lambda : \lambda \ident^E \otimes \sigma^{AK} \geqslant \omega^{AEK} \right\}\\
		&\leqslant \min_{\sigma^A} \min \left\{ \lambda : \lambda \ident^E \otimes \sigma^A \otimes \frac{\ident^K}{|K|} \geqslant \omega^{AEK} \right\}\\
		&= |K| \min_{\sigma^A} \min \left\{ \lambda : \lambda \ident^{EK} \otimes \sigma^A \geqslant \omega^{AEK} \right\}\\
		&\leqslant |K| \min_{\sigma^A} \min \left\{ \lambda : \lambda \ident^{E} \otimes \sigma^A \geqslant \omega^{AE} \right\}\\
		&= |K| 2^{-H_{\infty}(E|A)}
	\end{align*}
	where the second inequality holds due to the fact that
	\begin{align*}
		&\omega^{AE} \leqslant \lambda \ident^E \otimes \sigma^A\\
		&\Rightarrow \omega^{AE} \otimes \ident^K \leqslant \lambda \ident^{EK} \otimes \sigma^A\\
		&\Rightarrow \omega^{AEK} \leqslant \lambda \ident^{EK} \otimes \sigma^A.
	\end{align*}
	The last implication is true since the classicality of $K$ ensures that $\omega^{AEK} \leqslant \omega^{AE} \otimes \ident^K$.
\end{proof}

\begin{proof}[Proof of Theorem \ref{thm:main-lower-bound}]
	Let $\mathcal{H}_{\tilde{A}} \cong \mathcal{H}_{E}$, and let $\rho^{\tilde{A}E} = \ketbra{\Phi^{+}}^{\tilde{A}E}$. Then, by the Fuchs-van de Graaf inequalities \cite{fuchs-vdG} \footnote{See also Equation 9.110 in Nielsen and Chuang}, we have that
	\begin{equation}
		F \left( \mathcal{E}(\rho), \Omega^A \otimes \frac{\ident^E}{2^n} \right)^2 \geqslant 1-\varepsilon
	\end{equation}

	Now, let $\zeta^{AEK}$ be a state such that $\tr_K[\zeta^{AEK}] = \mathcal{E}(\rho)$ and in which the $K$ register holds the key:
\[ \zeta^{AEK} = \frac{1}{|K|} \sum_k \mathcal{E}_k(\rho^{\tilde{A}E}) \otimes \ketbra{k}^K. \]
Then, by Uhlmann's theorem (\cite{uhlmann}, or see Theorem 9.4 in \cite{NC2000}),
	\begin{equation*}
		F\left( \mathcal{E}(\rho), \Omega \otimes \frac{\ident}{2^n} \right)^2 = \max_{\sigma, \tr_K[\sigma] = \Omega \otimes \frac{\ident}{2^n}} F\left( \zeta^{AEK}, \sigma^{AEK} \right)^2
	\end{equation*}
	Now, let $\sigma^{AEK}$ be a state such that $\tr_K[\sigma^{AEK}] = \Omega \otimes \frac{\ident}{2^n}$ that maximizes the above fidelity. Also, let $\omega^{AEKK'} = V \sigma^{AEK} V\mdag$ and $\xi^{AEKK'} = V\zeta^{AEK}V\mdag$, where $V^{K \rightarrow K K'} = \sum_k \ket{kk}\bra{k}$, $\mathcal{H}_K \cong \mathcal{H}_{K'}$ and $\{ \ket{k} \}_{k \in K}$ is the computational basis on $\mathcal{H}_K$. Note that this ensures that $\omega^{AEK}$ is classical on $K$. We then have:
	\begin{align*}
		F\left( \mathcal{E}(\rho), \Omega \otimes \frac{\ident}{2^n} \right)^2 &= F\left( \zeta^{AEK}, \sigma^{AEK} \right)^2\\
		&= F\left( \xi^{AEKK'}, \omega^{AEKK'} \right)^2\\
		&\leqslant F\left( \xi^{AEK}, \omega^{AEK} \right)^2\\
		&\leqslant F\left( \Phi^{\tilde{A}E}, \mathcal{D}(\omega^{AEK}) \right)^2\\
		&\leqslant \max_{\mathcal{G}^{AK \rightarrow \tilde{A}}} F\left( \Phi^{\tilde{A}E}, \mathcal{G}(\omega^{AEK}) \right)^2\\
		&= \frac{2^{-H_{\infty}(E|AK)_{\omega}}}{2^n}
	\end{align*}
	where $\mathcal{D}^{AK \rightarrow \tilde{A}}$ is a superoperator which decrypts and then forgets the key. Now, by Lemma \ref{lem:lower-bound-min-entropy-omega} above, we have that for any state $\omega^{AEK}$ such that $\tr_K[\omega^{AEK}] =  \Omega^A \otimes \frac{\ident^E}{2^n}$ that is classical on $K$, $H_{\infty}(E|AK)_{\omega} \geqslant n - \log |K|$. Hence,
	\begin{align}
		1-\varepsilon &\leqslant F\left( \mathcal{E}(\rho), \Omega \otimes \frac{\ident}{2^n} \right)^2\\
		&\leqslant \frac{2^{-n + \log |K|}}{2^n}\\
		&= |K| \cdot 2^{-2n}
	\end{align}
	and therefore $|K| \geqslant 2^{2n}(1-\varepsilon)$. Hence, $\log|K| \geqslant 2n + \log(1-\varepsilon) \geqslant 2n - 1$ if, as assumed, $\varepsilon \leqslant \frac{1}{2}$.
\end{proof}
The tighter bound of \cite{DS2004} for schemes using public coins, given there as proposition 3.8, cannot be similarly generalized.

\section{Conclusion}
We have shown how to fully generalize the notions of entropic security and entropic indistinguishability without making any assumption on the entanglement between the sender and the adversary.  Furthermore, we proved that the two approximate quantum encryption scheme presented in \cite{AS2004} are also secure in this model. Is it possible to prove a general theorem showing that every quantum encryption scheme is entropically secure? If it is true, it would require different techniques than the ones used here, since our proofs rely on the fact that the ciphers give guarantees in the 2-norm, and not only in the 1-norm as in the definition of an approximate cipher. We leave this as an open problem.

\section*{Acknowledgments}
The authors would like to thank the following people for enlightening discussions and/or useful comments on the draft of this paper: Gilles Brassard, Claude Cr\'epeau, Patrick Hayden, Debbie Leung, Jean-Raymond Simard and Adam Smith. They would also like to thank the referree for pointing out a flaw in an earlier version of this paper. The authors would also like to acknowledge QuantumWorks, CIFAR and NSERC for their support.

\bibliographystyle{IEEEtran}
\bibliography{IEEEabrv,Quantum_Entropic_Security}

\end{document}